\documentclass[a4paper,11pt]{elsarticle}
\usepackage[english]{babel}
\usepackage[T1]{fontenc}
\usepackage{amsmath, amsthm}
\usepackage{ae,aecompl}
\usepackage{amssymb}
\usepackage[section] {placeins}
\usepackage[utf8]{inputenc}

\usepackage[ruled, vlined, linesnumbered]{algorithm2e}
\newtheorem {Theorem}                 {Theorem}         [section]
\newtheorem {theorem}      [Theorem]  {Theorem}
\newtheorem {myalgorithm}    [Theorem]  {Algorithm}

\newtheorem {lemma}        [Theorem]  {Lemma}

\newtheorem {problem}  [Theorem]  {Problem}

\input amssym.def
\input amssym

\usepackage{pgfplots}
\usepackage{verbatim}
\usepackage{subfigure}
\usepackage{tikz}
\usetikzlibrary{shapes,arrows,backgrounds,mindmap}
\usepackage{titlesec}
\journal{arXiv}

\begin{document}
	\begin{frontmatter}
		\title{Problems related to strong connectivity and strong biconnectivity}
		\author{Raed Jaberi}
		\ead{Jaberi.Raed@gmail.com}

		\begin{abstract}
		
Let $G=(V,E)$ be a strong biconnected graph and let $B \subseteq V$ such that for each vertex $w \in B$, the subgraph $G \setminus \lbrace w\rbrace$ is strongly connected. In this paper we study the problem of computing a subset $E_{\beta} \subseteq E$ of minimum size such that the subgraph $G_{\beta}=(V,E_{\beta})$ is strongly biconnected and for each vertex $w \in B$, the subgraph $G_{\beta} \setminus \lbrace w\rbrace$ is strongly connected. We prove that there exists a polynomial time $7$-approximation algorithm for this problem.

	\end{abstract} 
	\end{frontmatter}
	\section{Introduction}
Let $G=(V,E)$ be a strong biconnected graph and let $B \subseteq V$ such that for each vertex $w \in B$, the subgraph $G \setminus \lbrace w\rbrace$ is strongly connected. In this paper we study the problem of computing a subset $E_{\beta} \subseteq E$ of minimum size such that the subgraph $G_{\beta}=(V,E_{\beta})$ is strongly biconnected and for each vertex $w \in B$, the subgraph $G_{\beta} \setminus \lbrace w\rbrace$ is strongly connected. Note that the minimum strongly biconnected spanning subgraph problem is a special case of this problem when $B=\emptyset$ \cite{Jaberi22}. Therefore, this problem is NP-complete.

The problem of calculating a strongly connected spanning subgraph with minimum number of edges is NP-complete \cite{Gareyetal1979,Khulleretal94}. In $1994$,
Khuller et al.\cite{Khulleretal94} gave an approximation algorithm for the minimum strongly connected spanning subgraph problem. In $2001$, Vetta \cite{Vetta2001} presented a polynomial time  approximation algorithms of an approximation ratio of $3/2$. Moreover, in $2003$, Zhao et al. \cite{Zhaoetal03} developed a linear time $5/3$-approximation for the minimum strongly-connected spanning subgraph problem. The problem of finding strongly connected components in directed graph has been thoroughly investigated and solved by efficient linear time algorithms 
\cite{TAARJAN72,Ahoetal99,Sharir1981,Cormen2022,Gabow2000,CM96, SandersDietzfelbingerDementiev2019,Nuutila1994,Pearce2016}.
Wu and Grumbach \cite{WG2010} introduced the concept of strongly biconnected directed graph and strongly biconnected components. The problem of computing a strongly biconnected spanning subgraph with minimum number of edges is NP-complete and has been studied by Jaberi in $2022$ \cite{Jaberi22}. Many concepts and problems related to strong connectivity, strong biconnectivity and twinless strong connectivity in directed graphs  have received considerable attention in 
\cite{WG2010,Botea2018,BoteaIJCAI2018,Botea2015,Raghavan06,Jaberi2019, JaberiComputing2twinlessblocks,Jaberi2021,Jaberi2022,Jaberi01897,GeorgiadisandKosinas20,Jaberi09793,
Jaberi03788,JaberiMansour2026,GevigneySzigeti2018,GeorgiadisKosinas2020}. $k$-vertex-connected and $k$-edge connected graphs have many useful properties \cite{Menger1927,Mader1978,Mader1985}. Testing wheteher undirected graph is $2$-vertex connected and finding bridges, articulation points and biconnected components can be done in linear time \cite{TAARJAN72,Gabow2000,Schmidt2013,Brandeseatal2002,Ebert1983,Tarjan1974}. Testing whether a directed graph is $2$-vertex connected can be done in linear time \cite{Georgiadis2010}. 
Strong articulation points and strong bridges can be identified in linear time using efficient algorithms presented by Italiano et al. \cite{ItalianoLauraSantaroni2012,Italiano2010,FirmaniGeorgiadisItalianoLauraSantaroni2016}. These efficient algorithms are based on a close relationship between strong articulation points, strong bridges and dominators in flowgraphs. The problem of finding dominators can be solved efficiently in flowgraphs \cite{BGKRTW00,GT16,GeorgiadisTarjan2005,LT79,Alstrupetal99,Buchsbaumetal200}. Carmesin et al. studied $k$-blocks in undirected graphs \cite{CDHH2014}. Makino \cite{Makino1988} and Reif et al. \cite{ReifSpirakis1981}  studied $k$-components in directed graphs.
The problem of computing $2$-connected components and $2$-blocks in directed graphs have been extensively studied in  \cite{Erusalimskiietal80,Jaberi2015,Jaberi2016,Jaberi09793,Jaberi2edgetwinlessblocks,JaberiComputing2twinlessblocks,GeorgiadisGeorgiadisLauraParotsidis2015,GeorgiadisItalianoLauraParotsidis2016,
GeorgiadisItalianoLauraParotsidis15,GeorgiadisItalianoLauraParotsidis2018,
GeorgiadisItalianoLauraParotsidisKaranasiouPaudel2018,GeorgiadisItalianoPapadopoulosPapadopoulosParotsidis2015, Henzingeretal2015,Dahiphaleetal2023,CHILP18}.

Let $G=(V,E)$ be a strong biconnected graph and let $B \subseteq V$ such that for each vertex $w \in B$, the subgraph $G \setminus \lbrace w\rbrace$ is strongly connected. In this paper we study the problem of computing a subset $E_{\beta} \subseteq E$ of minimum size such that the subgraph $G_{\beta}=(V,E_{\beta})$ is strongly biconnected and for each vertex $w \in B$, the subgraph $G_{\beta} \setminus \lbrace w\rbrace$ is strongly connected. We will denote this problem by MBSC. We prove that there exists a polynomial time $7$-approximation algorithm for MBSC
\section{A special case where  all cervices in $V\setminus B$ are strong articulation points}
In $2015$, Jaberi \cite{Jaberi2015} proved that there is an approximation algorithm with approximation factor $17/3$ for the minimum strongly connected spanning subgraph with the same strong articulation points. His proof is based on the results of Georgiadis \cite{Georgiadis2010}, Italiano \cite{ItalianoLauraSantaroni2012}, and Zhao et al. \cite{Zhaoetal03}.

Example \ref{figure:example1section1} shows that an optimal solution for the minimum strongly connected spanning subgraph with the same strong articulation points is not necessarily an optimal solution for MBSC.  
\begin{figure}[htbp]
	\centering

\subfigure[]{	
\scalebox{0.91}{

		\begin{tikzpicture}[xscale=1.7]
		\tikzstyle{every node}=[color=black,draw,circle,minimum size=1cm]
			\node (v8) at (0.2,3.4) {$8$};
		\node (v9) at (-0.71,0.1) {$9$};
	
		\node (v1) at (-1.75,3.1) {$1$};
		\node (v10) at (1.72,-1.5) {$10$};
		\node (v2) at (-2.31,0) {$2$};
		\node (v3) at (0, -2.45) {$3$};
		\node (v4) at (3.43,-2) {$4$};
		\node (v5) at (1.41,0.5) {$5$};
		\node (v6) at (3.71,0.6) {$6$};
		\node (v7) at (3.72,3.3) {$7$};

		\begin{scope}   
		\tikzstyle{every node}=[auto=right] 
		\draw [-triangle 45] (v3) to (v5);
	    \draw [-triangle 45] (v5) to (v10);
	    \draw [-triangle 45] (v10) to (v4);
	    \draw [-triangle 45] (v4) to (v6);
	    \draw [-triangle 45] (v6) to (v7);
	    \draw [-triangle 45] (v7) to (v8);
	    \draw [-triangle 45] (v8) to (v1);
	    \draw [-triangle 45] (v1) to (v2);
	    \draw [-triangle 45] (v2) to (v9);
	    \draw [-triangle 45] (v9) to (v3);
	    \draw [-triangle 45] (v9) to (v5);
	    \draw [-triangle 45] (v9) to (v1);
	    \draw [-triangle 45] (v9) to (v8);
	    \draw [-triangle 45] (v2) to (v3);
	   
	    \draw [-triangle 45] (v5) to (v8);
	  \draw [-triangle 45] (v10) to (v6);
		\end{scope}
		\end{tikzpicture}
	}
	}	
	\subfigure[]{	
\scalebox{0.91}{

		\begin{tikzpicture}[xscale=1.7]
		\tikzstyle{every node}=[color=black,draw,circle,minimum size=1cm]
			\node (v8) at (0.2,3.4) {$8$};
		\node (v9) at (-0.71,0.1) {$9$};
	
		\node (v1) at (-1.75,3.1) {$1$};
		\node (v10) at (1.72,-1.5) {$10$};
		\node (v2) at (-2.31,0) {$2$};
		\node (v3) at (0, -2.45) {$3$};
		\node (v4) at (3.43,-2) {$4$};
		\node (v5) at (1.41,0.5) {$5$};
		\node (v6) at (3.71,0.6) {$6$};
		\node (v7) at (3.72,3.3) {$7$};

		\begin{scope}   
		\tikzstyle{every node}=[auto=right] 
		\draw [-triangle 45] (v3) to (v5);
	    \draw [-triangle 45] (v5) to (v10);
	    \draw [-triangle 45] (v10) to (v4);
	    \draw [-triangle 45] (v4) to (v6);
	    \draw [-triangle 45] (v6) to (v7);
	    \draw [-triangle 45] (v7) to (v8);
	    \draw [-triangle 45] (v8) to (v1);
	    \draw [-triangle 45] (v1) to (v2);
	    \draw [-triangle 45] (v2) to (v9);
	   
	    \draw [-triangle 45] (v9) to (v5);
         \draw [-triangle 45] (v2) to (v3);
	 \draw [-triangle 45] (v10) to (v6);

		\end{scope}
		\end{tikzpicture}
	}
	}	
	\subfigure[]{	
\scalebox{0.91}{

		\begin{tikzpicture}[xscale=1.7]
		\tikzstyle{every node}=[color=black,draw,circle,minimum size=1cm]
			\node (v8) at (0.2,3.4) {$8$};
		\node (v9) at (-0.71,0.1) {$9$};
	
		\node (v1) at (-1.75,3.1) {$1$};
		\node (v10) at (1.72,-1.5) {$10$};
		\node (v2) at (-2.31,0) {$2$};
		\node (v3) at (0, -2.45) {$3$};
		\node (v4) at (3.43,-2) {$4$};
		\node (v5) at (1.41,0.5) {$5$};
		\node (v6) at (3.71,0.6) {$6$};
		\node (v7) at (3.72,3.3) {$7$};

		\begin{scope}   
		\tikzstyle{every node}=[auto=right] 
		\draw [-triangle 45] (v3) to (v5);
	    \draw [-triangle 45] (v5) to (v10);
	    \draw [-triangle 45] (v10) to (v4);
	    \draw [-triangle 45] (v4) to (v6);
	    \draw [-triangle 45] (v6) to (v7);
	    \draw [-triangle 45] (v7) to (v8);
	    \draw [-triangle 45] (v8) to (v1);
	    \draw [-triangle 45] (v1) to (v2);
	    \draw [-triangle 45] (v2) to (v9);
	  
	    \draw [-triangle 45] (v9) to (v3);
         \draw [-triangle 45] (v2) to (v3);  
	     
		\end{scope}
		\end{tikzpicture}
	}
	}		
	
\caption{(a) Each vertex in $V\setminus \lbrace 3,9,6 \rbrace$ is a strong articulation point. (b) An optimal solution for minimum strongly connected spanning subgraph problem with same strong articulation points. (c) An optimal solution for MBSC when $B=\lbrace 9 \rbrace$.}
\label{figure:example1section1}
\end{figure}
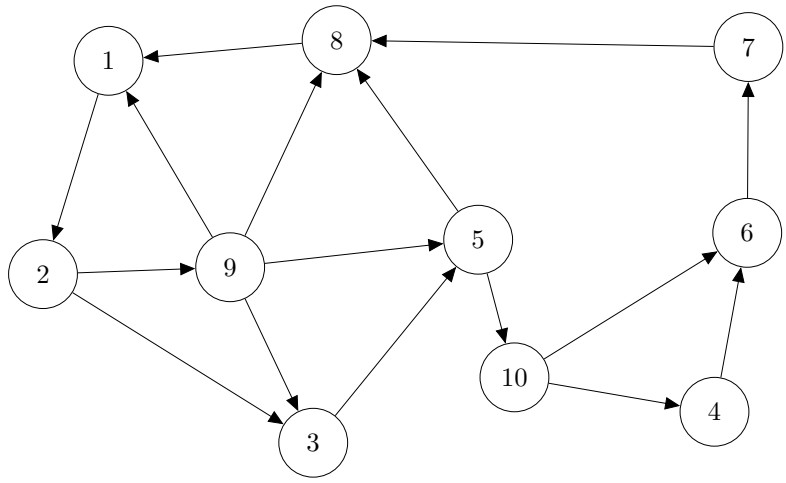
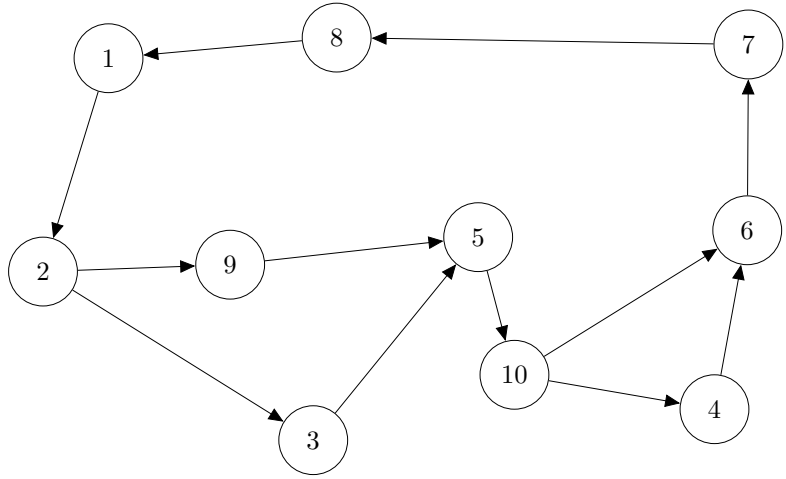
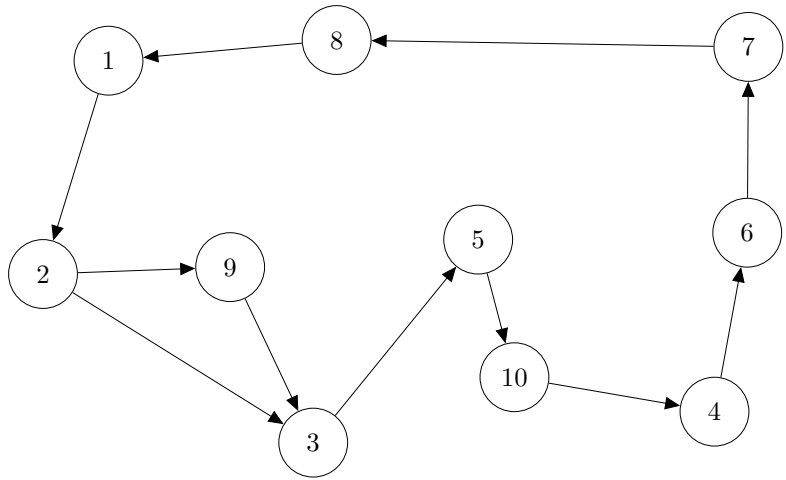
Moreover, Example \ref{figure:example2section1} shows that an optimal solution for the minimum strongly connected spanning subgraph problem is not necessarily an optimal solution for MBSC. 

\begin{figure}[htbp]
	\centering

\subfigure[]{	
\scalebox{0.91}{

		\begin{tikzpicture}[xscale=1.7]
		\tikzstyle{every node}=[color=black,draw,circle,minimum size=1cm]
			\node (v8) at (0.2,3.4) {$8$};
		\node (v9) at (-0.71,0.1) {$9$};
	
		\node (v1) at (-1.75,3.1) {$1$};
		\node (v10) at (1.72,-1.5) {$10$};
		\node (v2) at (-2.31,0) {$2$};
		\node (v3) at (0, -2.45) {$3$};
		\node (v4) at (3.43,-2) {$4$};
		\node (v5) at (1.41,0.5) {$5$};
		\node (v6) at (3.71,0.6) {$6$};
		\node (v7) at (3.72,3.3) {$7$};

		\begin{scope}   
		\tikzstyle{every node}=[auto=right] 
		\draw [-triangle 45] (v3) to (v5);
	    \draw [-triangle 45] (v5) to (v10);
	    \draw [-triangle 45] (v10) to (v4);
	    \draw [-triangle 45] (v4) to (v6);
	    \draw [-triangle 45] (v6) to (v7);
	    \draw [-triangle 45] (v7) to (v8);
	    \draw [-triangle 45] (v8) to (v1);
	    \draw [-triangle 45] (v1) to (v2);
	    \draw [-triangle 45] (v2) to (v9);
	    \draw [-triangle 45] (v9) to (v3);
	    \draw [-triangle 45] (v9) to (v5);
	    \draw [-triangle 45] (v9) to (v1);
	    \draw [-triangle 45] (v9) to (v8);
	    \draw [-triangle 45] (v2) to (v3);
	   
	    \draw [-triangle 45] (v5) to (v8);
	  
		\end{scope}
		\end{tikzpicture}
	}
	}	
	\subfigure[]{	
\scalebox{0.91}{

		\begin{tikzpicture}[xscale=1.7]
		\tikzstyle{every node}=[color=black,draw,circle,minimum size=1cm]
			\node (v8) at (0.2,3.4) {$8$};
		\node (v9) at (-0.71,0.1) {$9$};
	
		\node (v1) at (-1.75,3.1) {$1$};
		\node (v10) at (1.72,-1.5) {$10$};
		\node (v2) at (-2.31,0) {$2$};
		\node (v3) at (0, -2.45) {$3$};
		\node (v4) at (3.43,-2) {$4$};
		\node (v5) at (1.41,0.5) {$5$};
		\node (v6) at (3.71,0.6) {$6$};
		\node (v7) at (3.72,3.3) {$7$};

		\begin{scope}   
		\tikzstyle{every node}=[auto=right] 
		\draw [-triangle 45] (v3) to (v5);
	    \draw [-triangle 45] (v5) to (v10);
	    \draw [-triangle 45] (v10) to (v4);
	    \draw [-triangle 45] (v4) to (v6);
	    \draw [-triangle 45] (v6) to (v7);
	    \draw [-triangle 45] (v7) to (v8);
	    \draw [-triangle 45] (v8) to (v1);
	    \draw [-triangle 45] (v1) to (v2);
	    \draw [-triangle 45] (v2) to (v9);
	    \draw [-triangle 45] (v9) to (v3);

		\end{scope}
		\end{tikzpicture}
	}
	}	
	\subfigure[]{	
\scalebox{0.91}{

		\begin{tikzpicture}[xscale=1.7]
		\tikzstyle{every node}=[color=black,draw,circle,minimum size=1cm]
			\node (v8) at (0.2,3.4) {$8$};
		\node (v9) at (-0.71,0.1) {$9$};
	
		\node (v1) at (-1.75,3.1) {$1$};
		\node (v10) at (1.72,-1.5) {$10$};
		\node (v2) at (-2.31,0) {$2$};
		\node (v3) at (0, -2.45) {$3$};
		\node (v4) at (3.43,-2) {$4$};
		\node (v5) at (1.41,0.5) {$5$};
		\node (v6) at (3.71,0.6) {$6$};
		\node (v7) at (3.72,3.3) {$7$};

		\begin{scope}   
		\tikzstyle{every node}=[auto=right] 
		\draw [-triangle 45] (v3) to (v5);
	    \draw [-triangle 45] (v5) to (v10);
	    \draw [-triangle 45] (v10) to (v4);
	    \draw [-triangle 45] (v4) to (v6);
	    \draw [-triangle 45] (v6) to (v7);
	    \draw [-triangle 45] (v7) to (v8);
	    \draw [-triangle 45] (v8) to (v1);
	    \draw [-triangle 45] (v1) to (v2);
	    \draw [-triangle 45] (v2) to (v9);
	  
	    \draw [-triangle 45] (v9) to (v5);
         \draw [-triangle 45] (v9) to (v3);  
	     
		\end{scope}
		\end{tikzpicture}
	}
	}		
	
\caption{(a) Each vertex in $V\setminus \lbrace 3,9 \rbrace$ is a strong articulation point. (b) An optimal solution for minimum strongly connected spanning subgraph problem.(c) An optimal solution when $B=\lbrace 3 \rbrace$.}
\label{figure:example2section1}
\end{figure}
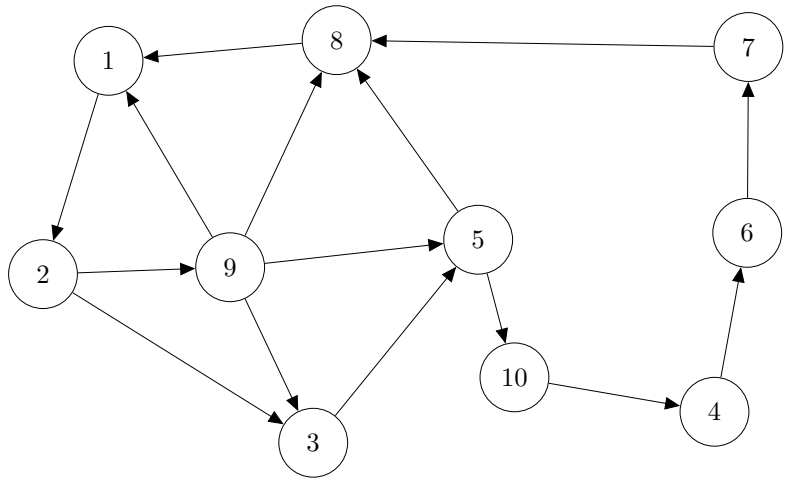
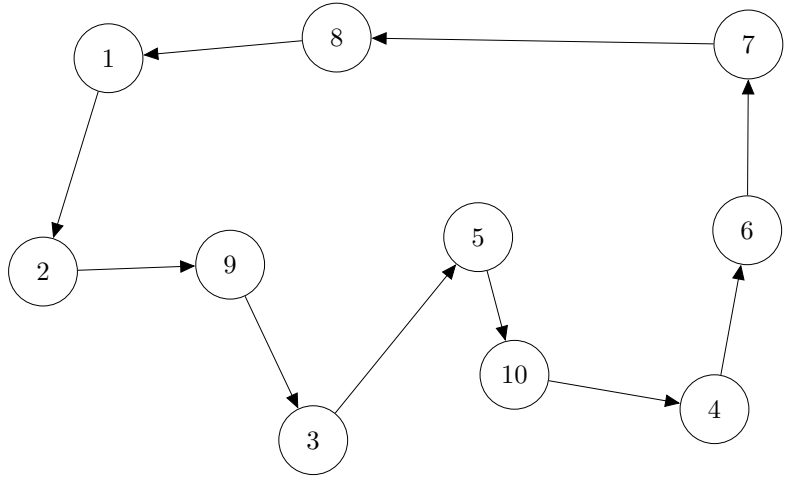
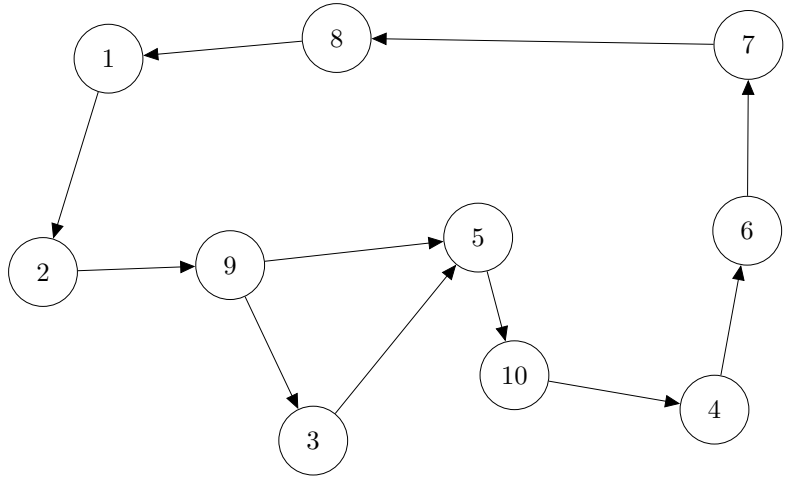
Let $G=(V,E)$ be a strong biconnected graph and let $B \subseteq V$ such that for each vertex $w \in B$, the subgraph $G \setminus \lbrace w\rbrace$ is strongly connected. If each vertex in $V\setminus B$ is a strong articulation point, then by results of Jaberi \cite{Jaberi2015}, there is an approximation algorithm with approximation factor $17/3$ for MBSC in this special case.

\section{Approximation algorithm for MBSC}
In this section, we prove that there is a $7$-approximation algorithm for MBSC. Algorithm \ref{algo:approxmiationalgorithmMBSC} is an approixmation algorithm for MBSC. This algorithms is based on \cite[Theorem $5.2$]{ItalianoLauraSantaroni2012} ,independent spanning trees \cite{GeorgiadisTarjan2005}, and the results of Georgiadis et al. \cite{GeorgiadisItalianoKaranasiou2020,GeorgiadisItalianoKaranasiou2017,Georgiadis12011}. 

\begin{figure}[htp]
	\begin{myalgorithm}\label{algo:approxmiationalgorithmMBSC}\rm\quad\\[-4ex]
		\begin{tabbing}
			\quad\quad\=\quad\=\quad\=\quad\=\quad\=\quad\=\quad\=\quad\=\quad\=\kill
			\textbf{Input:} A strongly biconnected graph $G=(V,E),B \subseteq V$ such that for each \\
			\>\>\>\> vertex $w \in B$, the subgraph $G \setminus \lbrace w\rbrace$ is strongly connected.\\
			\textbf{Output:} A strongly biconnected subgraph $G_{\beta}=(V,E_{\beta})$ such that for each \\
			 \>\>\>\> vertex $w \in B$, the subgraph $G_{\beta} \setminus \lbrace w\rbrace$ is strongly connected.\\
			{\small 1}\> \textbf{if} $B = V$ \textbf{then}\\
            {\small 2}\>\> $E_{1}\leftarrow$ the edge set of a minimal $2$-vertex-connected spanning subgraph of $G$ \\
            	{\small 3}\>\> $G_{\beta}\leftarrow(V,E_{\beta} )$, where $E_{\beta}=E_{1}$\\
             {\small 4}\> \textbf{else} \textbf{if} $B$ has no vertices \textbf{then} \\
             {\small 5}\>\> $G_{\beta}\leftarrow (V,E_{\beta} )$, where $E_{\beta} $ is the edge set of a minimal strongly biconnected subgraph of $G$ \\
             {\small 6}\> \textbf{else}\\		
             {\small 7}\>\>  $E_{1}\leftarrow$ the edge set of a minimal strongly biconnected spanning subgraph of $G$ \\
            {\small 8}\>\> \textbf{if} $|B|=1$ \textbf{then}\\
            {\small 9}\>\>\>\> Let $s$ be the vertex in $B$\\
             {\small 10}\>\>\>\> Let $r$ be a vertex in $V\setminus B$\\
            {\small 11}\>\>\>\>  $T \leftarrow$ a spanning tree rooted at $r$ of $G\setminus \lbrace s\rbrace$\\
            {\small 12}\>\>\>\>  $T^{r} \leftarrow$ a spanning tree rooted at $r$ in the reversal graph of $G\setminus \lbrace s\rbrace$\\
            {\small 13}\>\>\>\> $E_{2}\leftarrow $ the edge set of $T$\\
            {\small 14}\>\>\>\>  \textbf{for} each edge $(v,w) \in T^{r}$ \textbf{do} \\         
            {\small 15}\>\>\>\>\>  $E_{2} \rightarrow E_{2}\cup \lbrace (w,v) \rbrace$\\
   
            {\small 16}\>\>\>\>  $G_{\beta}\leftarrow(V,E_{\beta} )$, where $E_{\beta}=E_{1}\cup E_{2}$\\     
             {\small 17}\>\> \textbf{else}\\
			{\small 18}\>\>\>\> Let $u \in V\setminus B$.\\
			{\small 19}\>\>\>\> $E_{2} \leftarrow $  the edge set of two independent trees rooted at $u$ of $G$\\
			{\small 20}\>\>\>\> $E_{3} \leftarrow \emptyset$\\
			
			{\small 21}\>\>\>\>  $E_{4} \leftarrow $  the edge set of two independent trees rooted at $u$ of the revesal graph $G^{R}$\\
			{\small 22}\>\>\>\>\textbf{for} each edge $(v,w) \in E_{4}$ \textbf{do} \\  
			{\small 23}\>\>\>\> \>\>$E_{3} \leftarrow E_{3}\cup \lbrace (w,v) \rbrace$\\
			{\small 24}\>\>\> \> $G_{\beta}\leftarrow(V,E_{\beta} )$, where $E_{\beta}=E_{1}\cup E_{2}\cup E_{3}$\\
			 {\small 25}\> \textbf{for} each edge $(u,w) \in E_{\beta}$ \textbf{do} \\
            {\small 26}\>\> \textbf{if} no vertex in $B$ is a strong articulation point of the subgraph $(V,E_{\beta}\setminus \lbrace (u,w) \rbrace)$  \textbf{then} \\
            {\small 27}\>\>\>$E_{\beta}\leftarrow E_{\beta}\setminus \lbrace (u,w) \rbrace$.\\
		{\small 28}\>\textbf{output} the subgraph $G_{\beta}=(V,E_{\beta})$ 
	
		\end{tabbing}
	\end{myalgorithm}
\end{figure}

Lemma \ref{def:feasiblesolutionproof} shows that the output of Algorithm \ref{algo:approxmiationalgorithmMBSC} is a feasible solution for MBSC.

\begin{lemma} \label{def:feasiblesolutionproof}
Let $G=(V,E)$ be a strong biconnected graph and let $B \subseteq V$ such that for each vertex $w \in B$, the subgraph $G \setminus \lbrace w\rbrace$ is strongly connected. Let $G_{\beta}=(V,E_{\beta} )$ be the output of the Algorithm \ref{algo:approxmiationalgorithmMBSC}. Then, $G_{\beta}=(V,E_{\beta} )$ is strongly biconnected and for each vertex $w \in B$, $w$ is not a strong articulation point.
\end{lemma}
\begin{proof}
 We consider three cases:
\begin{itemize} 
\item Case $1$. $B=V$. In this case, $G_{\beta}=(V,E_{\beta} )$ is a minimal $2$-vertex connected spanning subgraph of $G$. Since the subgraph $G_{\beta}=(V,E_{\beta} )$ is $2$-vertex connected, $G_{\beta}=(V,E_{\beta} )$ has no strong articulation points. Moreover, $G_{\beta}=(V,E_{\beta} )$ is strongly biconnected because its underlying graph is biconnected.
\item Case $2$. $|B|=0$. In this case. each feasible solution for minimum strongly biconnected spanning subgrpah problem is a feasible solution for MBSC beacuse minimum strongly biconnected spanning subgraph problem is a special case of MBSC when $B$ does not contain any vertex..
\item Case $3$. $|B|>0 $ and $|B|<n$. Since  $E_{1}\subseteq E_{\beta}$,the subgraph  $G_{\beta}=(V,E_{\beta} )$ is strongly biconnected. Now we need to ensure that none of the vertices in $B$ is a strong articulation point of the subgraph $G_{\beta}=(V,E_{\beta} )$. 
We consider two subcases:
\begin{enumerate}
\item Case $3(a)$. $|B|=1$ and $B \neq V$. Let $B=\lbrace s \rbrace$. The subgraph  $(V\setminus \lbrace s \rbrace,E_{2})$ is strongly connected and 
$E_{2}\subseteq E_{\beta}$. Therefore. the vertex $s$ is not a strong articulation point of $G_{\beta}=(V,E_{\beta} )$.
 
\item Case $3(b)$. $|B|>1 $ and $|B|\neq V$. By \cite[Theorem $5.2$]{ItalianoLauraSantaroni2012} and the properties of two independent tree \cite{Georgiadis12011,GeorgiadisTarjan2005}, each vertex in $B$ is not a strong articulation point of $G_{\beta}=(V,E_{\beta} )$.
\end{enumerate}
\end{itemize}

\end{proof} 

Let $E_{O}$ be an optimal solution for MBSC. Since the subgraph $(V,E_{O}) $ is strongly biconnected, the subgraph $(V,E_{O}) $ is strongly connected and hence $|E_{O}|\geq n$. The following theorem shows that Algorithm \ref{algo:approxmiationalgorithmMBSC} has an approximation ratio of $7$.
\begin{theorem} \label{def:theoremforapproximationratio}
Let $G_{\beta}=(V,E_{\beta} )$ be the output of the Algorithm \ref{algo:approxmiationalgorithmMBSC}. Then 
$|E_{\beta}|\leq 7n-4$.
\end{theorem}
\begin{proof}
Let $G=(V,E)$ be a strong biconnected graph and let $B \subseteq V$ such that for each vertex $w \in B$, the subgraph $G \setminus \lbrace w\rbrace$ is strongly connected. Let $G_{\beta}=(V,E_{\beta} )$ be the output of the Algorithm \ref{algo:approxmiationalgorithmMBSC}. We consider three cases:
\begin{itemize} 
\item Case $1$. $B=V$. In this case, the directed graph $G=(V,E)$ is $2$-vertex connected since no vertex in $B$ is a strong articulation point of $G$. The output $G_{\beta}=(V,E_{\beta})$ is a minimal $2$-vertex connected spanning subgraph of $G$. Results of Edmonds \cite{Edmonds1972} and Mader \cite{Mader1985} imply that  $|E_{\beta}| \leq 4n$ \cite{CheriyanTThurimella2000,Georgiadis12011}.
\item Case $2$. $B$ contains no vertices. In this case. each solution for minimum strongly biconnected spanning subgrpah problem is a solution for MBSC. Therefore, $|E_{\beta}| \leq 3n$ \cite{Jaberi2022}.
\item Case $3$. $|B|>0 $ and $|B|<n$. In this case we have $|E_{1}| \leq 3n$ \cite{Jaberi2022}.We consider two subcases:
\begin{enumerate}
\item Case $3(a)$. $|B|=1$ and $B\neq V$. Note that $E_{2}$ consists of two spanning trees. Thus, we $|E_{2}|\leq 2(n-2)$.
\item Case $3(b)$. $|B|>1 $ and $|B|\neq V$. Note that $E_{2}$ consists of two independent trees  and $E_{3}$ consists of two independent trees. Therefore, we have $|E_{2}\cup E_{3}|\leq 4n-4$.
\end{enumerate}
\end{itemize}

\end{proof}

The following theorem shows that Algorithm \ref{algo:approxmiationalgorithmMBSC} runs in polynomial time.
\begin{Theorem}
The running time of Algorithm \ref{algo:approxmiationalgorithmMBSC} is $O(nm)$.
\end{Theorem}
\begin{proof}
	if $B=V$, then the input graph is $2$-vertex connected. A minimal $2$-vertex-connected subgraph of a $2$-vertex connected directed graph can be found in time $O(n^2)$ \cite{Georgiadis12011,GeorgiadisItalianoKaranasiou2020}.
	Therefore, line $2$ takes time $O(n^2)$. Moreover, a minimal strongly biconnected spanning subgraph of a strongly biconnected directed graph can be computed in time $O(nm)$ \cite{Jaberi22}. A spanning tree in directed graphs can be found in linear time using DFS or BFS. Non tirival dominators and two independent spanning trees of a flow graph can be identified in $O(n+m)$ time. \cite{GeorgiadisTarjan2005,Georgiadis12011}.
\end{proof}
\section{Additional problems related to strong connectivity amd strong biconnectivity}
In this section we describe additional problems related to strong connectivity and strong biconnectivity.
\begin{problem} \label{def:problemmbsci}
Let $G=(V,E)$ be a strong biconnected graph, let $i$ be an integer, and let $B \subseteq V$ such that for each vertex $w \in B$, the subgraph $G \setminus \lbrace w\rbrace$ is strongly connected. The goal is to compute a subset $E_{\beta} \subseteq E$ of minimum size such that the subgraph $G_{\beta}=(V,E_{\beta})$ is strongly biconnected and $|\lbrace v\mid v\in B \text{ and } v \text{ not a strong articulation point of } G_{\beta} \rbrace|\geq i$.
\end{problem}
Problem \ref{algo:approxmiationalgorithmMBSC} is a special case of the minimum 2-vertex connected spanning subgraph problem when $B=V$ and $i=n$. Algorithm \ref{algo:approxmiationalgorithmMBSC} is a polynomial time approximation algorithm with approximation factor $7$ for this problem.

\begin{problem} \label{def:problemmbscatmosti}
Let $G=(V,E)$ be a $2$-vertex strongly biconnected graph \cite{Jaberi2021} and let $i$ be an integer.
The goal is to compute a subset $E_{\beta} \subseteq E$ of minimum size such that the subgraph $G_{\beta}=(V,E_{\beta})$ is strongly biconnected and the number of strong articulation point of this subgraph is at most $i$.
\end{problem}
This problem is NP-complete because the minimum $2$-vertex connected spanning subgraph problem is a special case of this problem when $i=0$. Note that an optimal solution for the minimum strongly connected spanning subgraph problem is not necessarily an optimal solution for this problem. Furthermore, an optimal solution for the minimum $2$-vertex strongly biconnected  spanning subgraph problem is not necessarily an optimal solution for this problem. 
Algorithm \ref{algo:approxmiationalgorithmMBSC} can be modified in order to obtain a polynomial time approximation algorithm with approximation factor $7$.

\begin{problem}\label{def:problembapi}
Let $G=(V,E)$ be a $2$-vertex strongly biconnected graph \cite{Jaberi2021} and let $i$ be an integer.
The goal is to compute a subset $E_{\beta} \subseteq E$ of minimum size such that the subgraph $G_{\beta}=(V,E_{\beta})$ is strongly biconnected and the number of  b-articulation point of this subgraph is at most $i$.
\end{problem}
The minimum $2$-vertex strongly biconnected spanning subgraph problem is a special case of Problem \ref{def:problembapi} when $i=0$. Therefore, Problem \ref{def:problembapi} is NP-complete.  Note that an optimal solution for the minimum strongly connected spanning subgraph problem is not necessarily an optimal solution for this problem. Furthermore, an optimal solution for the minimum $2$-vertex strongly biconnected  spanning subgraph problem is not necessarily an optimal solution for this problem.

\end{document}